\documentclass{article}

\usepackage[english]{babel}

\usepackage[a4paper,top=2cm,bottom=2cm,left=3cm,right=3cm,marginparwidth=1.75cm]{geometry}

\usepackage{amsmath}
\usepackage{graphicx}
\usepackage[colorlinks=true, allcolors=blue]{hyperref}
\usepackage{amsthm}

\usepackage{wasysym}

\newtheorem{theorem}{Theorem}

\newtheorem{corollary}[theorem]{Corollary}

\newcommand{\bra}[1]{\langle #1|}
\newcommand{\ket}[1]{|#1\rangle}
\newcommand{\braket}[2]{\langle #1|#2\rangle}
\renewcommand{\cent}[0]{\mbox{\textcent}}
\newcommand{\lmoon}[0]{\leftmoon\mspace{-3mu}}
\newcommand{\rmoon}[0]{\mspace{-3mu}\rightmoon}

\newcommand{\dollar}[0]{\$}

\newcommand{\mymatrix}[2]{\left( \begin{array}{#1} #2\end{array} \right)}

\newcommand{\myvector}[1]{\mymatrix{c}{#1}}

\newcommand{\mypar}[1]{\left( #1 \right) }

\newcommand{\tildesigma}{\widetilde{\Sigma}}
\newcommand{\tildegamma}{\widetilde{\Gamma}}

\newcommand{\tildex}{\tilde{x}}
\newcommand{\hatc}{\hat{c}}

\newcommand{\usquare}{\mathtt{USQUARE}}

\title{Classical and quantum Merlin-Arthur automata}
\author{Abuzer Yakary\i lmaz$^{1,2}$ \\ ~ \\
$^1$Center for Quantum Computer Science, Faculty of Computing, \\ University of Latvia, R\={\i}ga, Latvia \and $^2$QWorld Association, Tallinn, Estonia, \url{https://qworld.net}
\\ ~ \\
\textit{email: abuzer.yakaryilmaz@lu.lv}
} 

\begin{document}
\maketitle

\begin{abstract}
We introduce Merlin-Arthur (MA) automata where Merlin provides a certificate at the beginning of computation and it is scanned by Arthur before reading the input. We define Merlin-Arthur deterministic, probabilistic, and quantum finite state automata (resp., MA-DFAs, MA-PFAs, and MA-QFAs) and postselecting MA-PFAs and MA-QFAs (resp., MA-PostPFA and MA-PostQFA). We present several results using different certificate lengths.

We show that MA-DFAs use constant length certificates, and they are equivalent to multi-entry DFAs. Thus, they recognize all and only regular languages, but they can be exponential and polynomial state efficient over binary and unary languages, respectively. With sublinear length certificates, MA-PFAs can recognize several nonstochastic unary languages with cutpoint 1/2. With linear length certificates, MA-PostPFAs can recognize these nonstochastic unary languages with bounded error. With arbitrarily long certificates, bounded-error MA-PostPFAs can verify every unary decidable language. With sublinear length certificates, bounded-error MA-PostQFAs can verify several nonstochastic unary languages. With linear length certificates, they can verify every unary language and some NP-complete binary languages. With exponential length certificates, they can verify every binary language.
\end{abstract}

{
  \small  \textbf{\textit{Keywords:}} Interactive proof systems, finite automata, quantum computing, probabilistic computation, postselection
}

\section{Introduction}

Interactive proof systems (IPSs) are one of the fundamental concept in complexity theory by extending the classical proof system with randomness and interaction \cite{Bab85,GMR89}. An IPS for a language $L$ has two parties: a prover $ P $ and a verifier $V$. The verifier $V$ is a (resource bounded) probabilistic machine, and there is no computational limitations on the prover $P$ as well as any possible (potentially cheating) provers $P^*$. The verifier $V$ accepts the member strings of $L$ with high probability by interacting with the prover $P$, and it rejects any string not in $L$ with high probability while interacting with any possible prover $P^*$. An IPS is called private-coin if the verifier hides its probabilistic choices \cite{GMR89}. If the probabilistic choices are visible to the verifier, then it is called public-coin. Hiding probabilistic choices allows us to design stronger protocols \cite{Con93}. The public-coin IPSs are also known as Arthur-Merlin (AM) systems or games \cite{Bab85}, where Arthur represents the verifier and Merlin represents the prover. If Merlin provides her proof at once and at the beginning of computation, then it is called a Merlin-Arthur proof system.

Quantum interactive proof systems (QIPs) were introduced in \cite{Wat99A}, where the prover and verifier send quantum messages (states) to each other. Quantum versions of Arthur-Merlin and Merlin-Arthur proof systems (resp., QAM and QMA) were also defined in the literature \cite{AN02,KSV02,Wat09A}. In QMA proof systems, the verifier is provided with a quantum proof (state) on the given input. If the proof is classical, it is called as QCMA (or MQA) \cite{AN02,Wat99A}.  

In polynomial time, we know that QIPs, IPSs, AM proof systems are equivalent, and they verify all and only the languages in $\mathsf{PSPACE}$ \cite{GS86,Shamir92,JJUW10}. Polynomial time MA or QMA proof systems are believed to be weaker, and any language verified by them is in $\mathsf{PP}$ \cite{KW00}. The complexity class of polynomial time QMA proof systems has drawn special interests with several known complete problems \cite{Bookatz14}.

Constant space verifiers (two-way probabilistic finite state automaton (2PFA) verifiers) were introduced in \cite{DS92}. As the verifiers have only finite number of states for their memories, the communication with the prover happens symbol by symbol. Private-coin 2PFA verifiers are quite powerful, and for example they can verify every language in $\mathsf{E}$ (i.e., exponential time with linear exponent). On the other hand, public-coin verifiers are weaker, as any language verified by them is in $\mathsf{P}$, and  they cannot verify certain languages in $\mathsf{L}$ such as $\mathtt{PAL} = \{ w \in \{a,b\}^* \mid w = w^R \}$. 

Private-coin (one-way) PFA verifiers can verify any language recognized by multi-head deterministic finite automata even using only constant amount of randomness \cite{SY14B}. However public-coin PFA verifiers cannot verify any nonregular languages \cite{CHPW98}. PFA verifiers with postselecting capability were introduced in \cite{DY18A,DY18A-arXiv}.

Quantum finite automata (QFAs) verifiers have been examined in different set-ups \cite{NY04,NY09,Yak13C,Yam14,VY14,ZQG15,YSD16,SY17}. Related to our work, AM systems with two-way (classical head) QFA verifiers using rational-valued transitions can verify every language known to be verified by private-coin IPS with 2PFA verifiers as well as some NEXP-complete languages \cite{Yak13C}. If real-valued transitions are allowed, then they can verify every language including uncomputable ones \cite{SY17}.

MA systems with QFA verifiers were first introduced in \cite{VY14}, where the certificate provided by Merlin is placed on a parallel track of the input. We introduce a different model, which is similar to the model ``Automata that take advice'' \cite{DH95}: the certificate provided by Merlin is fed to the verifier before reading the input. We remark that the set-up in \cite{VY14} may also be seen as one-way private-coin IPSs where the certificate is read (synchronously or asynchronous) in parallel to reading the input (e.g., \cite{SY14,DY18A}).

We define Merlin-Arthur deterministic, probabilistic, and quantum finite state automata (resp., MA-DFAs, MA-PFAs, and MA-QFAs) and postselecting MA-PFAs and MA-QFAs (resp., MA-PostPFA and MA-PostQFA). We present several results where the certificate length for the member strings are sublinear, linear, exponential, and arbitrarily long.

We show that MA-DFAs can benefit from only the constant length certificates, and they are equivalent to multi-entry DFAs. Thus, they recognize all and only regular languages, but they can be exponential and polynomial state efficient over binary and unary languages, respectively. With sublinear length certificates, MA-PFAs can recognize several nonstochastic unary languages with cutpoint 1/2. With linear length certificates, MA-PostPFAs can recognize these nonstochastic unary languages with bounded error. With arbitrarily long certificates, bounded-error MA-PostPFAs can verify every unary decidable language. With sublinear length certificates, bounded-error MA-PostQFAs can verify several nonstochastic unary languages. With linear length certificates, they can verify every unary language and some NP-complete binary languages. With exponential length certificates, they can verify every binary language.

We assume the reader familiar with the basics of automata theory, complexity theory, and quantum computing (see e.g., \cite{NC00,Wat09A,Sip13,SY14,AY21}). We give the definitions and notations used throughout the paper in the next section. Then, we introduce our models in Section~\ref{sec:MA-def}. We present our results using sublinear, linear, exponential, and arbitrarily long length certificates in Sections~\ref{sec:sublinear}, \ref{sec:linear}, \ref{sec:exp}, and \ref{sec:arbitrary}, respectively. We close the paper with Section~\ref{sec:conc-remark}.

\section{Preliminaries}
\label{sec:pre}

We denote the input alphabet $\Sigma$ not containing the left and right end-markers ($\cent$ and $\dollar$, respectively); and $\tildesigma = \Sigma \cup \{\cent,\dollar\}$.  We denote the certificate alphabet $\Gamma$ not containing the symbols $ \lmoon $ and $\rmoon$, indicators for the beginning and end of the certificate (proof), respectively; and $\tildegamma = \Gamma \cup \{ \lmoon, \rmoon \}$. We denote the empty string $\varepsilon$ (its length is zero) and it is the same for all alphabets. For any given (non-empty) string $x$, $|x|$ is its length and $x_i$ represents its $i$-th symbol from the left.

Any given input $x \in \Sigma^*$ is read as $\tildex = \cent x \dollar $, from left to right and symbol by symbol. Any given certificate $c \in \Gamma^*$ is read as $\hatc = \lmoon c \rmoon$, from left to right and symbol by symbol. 

\subsection{Standard models}

An $m$-state probabilistic finite automaton \cite{Rab63,Paz71} (PFA) M is a 5-tuple
\[
    M = (\Sigma,S,\{ A_\sigma \mid \sigma \in \tildesigma\}, s_I, S_a ),
\]
where
\begin{itemize}
    \item $S = \{s_1,\ldots,s_m\}$ is the finite set of states,
    \item $A_\sigma$ is the (left) stochastic transition matrix for the symbol $\sigma \in \tildesigma$,
    \item $s_I \in S$ is the initial (starting) state, which is $s_1$ if not specified otherwise, and,
    \item $S_a \subseteq S$ is the set of accepting state(s).
\end{itemize}

The computation of $M$ is traced by an $m$-dimensional column vector called (probabilistic) state vector, which is a stochastic vector. Its $j$-th entry represents the probability of $M$ being in the state $s_j$, which is a non-negative real number. The entry summation of state vector must be 1. The computation starts in the state vector $v_0$ with $v_0[I]=1$ (and so the rest entries are zeros). 

Any given input $x \in \Sigma^*$ with length $l$, $M$ reads the input as
\[
    \cent~x_1~x_2~\ldots~x_l~\dollar. 
\]
Then, the computation evolves as
\[
v_f = A_{\dollar} A_{x_l} A_{x_{l-1}} \cdots A_{x_1} A_{\cent} v_0, 
\]
where $v_f$ is the final state vector. The input $x$ is accepted with probability
\[
    Acc_M(x) = \sum_{s_j \in S_a} v_f[j],
\]
and it is rejected with probability $Rej_M(x) = 1 - Acc_M(x)$.

A deterministic finite automaton (DFA) is a restricted PFA using only 0s and 1s as transition values. In the literature, the formal definition of DFAs usually does not include the end-markers. One may add or remove the end-markers by modifying the initial and accepting state(s) and without changing the number of states.

There are different variants of quantum finite automata (QFAs) in the literature \cite{AY21}. We use the most general one here. An $m$-state QFA is a 5-tuple
\[
    M = (\Sigma,Q,\{\mathcal{E}_\sigma \mid \sigma \in \tildesigma\},q_I,Q_a),
\]
where
\begin{itemize}
    \item $Q = \{q_1,\ldots,q_m\}$ is the finite set of states,
    \item $\mathcal{E}_\sigma = \{E_{\sigma,1},\ldots,E_{\sigma,k_\sigma}\}$ is the superoperator with $k_\sigma$ operation elements for the transitions of symbol $\sigma \in \tildesigma$,
    \item $q_I \in Q$ is the initial (starting) state, which is $q_1$ if not specified otherwise, and,
    \item $Q_a \subseteq Q$ is the set of accepting state(s).
\end{itemize}
Remark that $\mathcal{E} = \{E_1,\ldots,E_k\}$ is a superoperator if and only if $\sum_{j=1}^k E_j^\dagger E_j = I $, or,  alternatively, the columns of the following matrix form an orthonormal set
\[
    \begin{array}{|c|} 
    \hline  E_1  \\ \hline E_2 \\ \hline \vdots \\ \hline E_k \\ \hline
    \end{array} ~.
\]
When $\mathcal{E}$ is applied to a pure quantum state $\ket{v}$, it may create a mixture of pure states. That is, $\mathcal{E}(\ket{v})$ gives us $k$ unnormalized quantum states:
\[
    \left\{ \ket{\widetilde{v'_1}} = E_1 \ket{v} , \ket{\widetilde{v'_2}} = E_2 \ket{v} , \ldots, \ket{\widetilde{v'_k}} = E_k \ket{v}  \right\}.
\]
Here $ \ket{\widetilde{v'_j}} $ is obtained with probability $p_j = \braket{ \widetilde{v'_j} }{ \widetilde{v'_j} }$. Thus, after normalization,
\[
    \ket{ \widetilde{v'_j}} \rightarrow \ket{v'_j} = \frac{\ket{ \widetilde{v'_j}}}{\sqrt{p_j}},
\] and, we have a mixture of
\[
    \left\{  (p_1,\ket{v'_1}),(p_2,\ket{v'_2}),\ldots,(p_k,\ket{v'_k})  \right\} .
\]
A convenient way to represent such mixture is using density matrix:
\[
    \rho' = \sum_{j=1}^k p_j \ket{v'_j}\bra{v'_j}.
\]
Indeed, the overall calculation can be shorten as follows: We start in a density matrix $\rho = \ket{v}\bra{v}$, and then, after applying $\mathcal{E}$, we obtain the density matrix
\[
    \rho' = \mathcal{E}(\rho) = \sum_{j=1}^k E_j \rho E_j^\dagger.
\]
One nice property of density matrices is that its diagonal entries give the probabilities of being in the basis states ($\ket{q_1},\ldots,\ket{q_m}$), and so its trace (the sum of diagonal entries) is equal to 1, i.e., $tr(\rho)=tr(\rho')=1$. 

Let $x \in \Sigma^*$ be the given input with length $l$. Thus, $M$ applies $l+2$ superoperators,
\[
    \mathcal{E}_{\cent}, \mathcal{E}_{x_1}, \ldots, \mathcal{E}_{x_l}, \mathcal{E}_{\dollar}. 
\]
The computation starts in $\rho_0 = \ket{q_I}\bra{q_I}$. After reading the $j$-th symbol of $\tildex$:
\[
    \rho_j = \mathcal{E}_{\tildex_j} (\rho_{j-1}),
\]
where $1 \leq j \leq l+2$. Let $\rho_f$ be the final density matrix. Then, the input $x$ is accepted with probability
\[
    Acc_M(x) = \sum_{q_j \in Q_a} \rho_f[j,j],
\]
and, it is rejected with probability $ Rej_M(x) =  1 - Acc_M(x)$.

\subsection{Postselecting models}

Postselection is an ability to  make the decision based on a specified set of outcomes (post-selected), where the other outcomes (non-postselected) are discarded \cite{Aar05,GNY21}. A postselecting automaton or an automaton with postselection, say $M$, differs from the standard model by having three types of states: accepting, rejecting, and non-postselecting. Then, at the end of computation, it is discarded if the computation ends in a non-postselecting state. Thus, the decision is made solely based on the accepting and rejecting states. 

For any given input $x \in \Sigma^*$,
we denote the accepting, rejecting, and non-postselecting probabilities of $M$ on $x$ as $a_M(x)$, $r_M(x)$, and $n_M(x)$, respectively. Then, the input $x$ is accepted and rejected by $M$ with probabilities
\[
    Acc_M(x) = \dfrac{a_M(x)}{a_M(x)+r_M(x)} ~~\mbox{ and }~~
    Rej_M(x) = \dfrac{r_M(x)}{a_M(x)+r_M(x)},
\]
respectively. Remark that to be able to postselect it must be guaranteed that $a_M(x)+r_M(x) > 0 $.

A postselecting (realtime) automaton can be seen as a very restrict two-way automaton \cite{YS10B,YS11B,YS13A}: If the computation ends in a non-postselecting state, the computation is restarted from the initial configuration, a restricted version of sweeping automaton. It was also shown that the capability of postselecting is equivalent to send a classical bit through a Closed Timelike Curve \cite{SY12B}.

We denote postselecting PFA and QFA as PostPFA and PostQFA, respectively.

\subsection{Language recognition and classes}

The languages recognized by DFAs form the set of regular languages \cite{RS59}.

A language $L \subseteq \Sigma^*$ is said to be recognized by $M$ with cutpoint $\lambda \in [0,1)$ if and only if 
\begin{itemize}
    \item every $x \in L$ is accepted by $M$ with probability greater than $\lambda$ and
    \item every $x \notin L$ is accepted by $M$ with probability at most $\lambda$.
\end{itemize}

If we fix the language recognition mode of a PFA with cutpoint zero, then we obtain a nondeterministic finite automaton (NFA). Therefore, PFAs recognize all and only regular languages with cutpoint zero.

The class of languages recognized by PFAs with cutpoints is stochastic languages \cite{Rab63}. It was shown that QFAs  recognize all and only stochastic languages with cutpoints \cite{YS11A}. On the other hand, QFAs with cutpoint 0 recognize all and only exclusive stochastic languages \cite{YS10A}, a superset of regular languages and a proper subset of stochastic languages \cite{Paz71}.

A language $L \subseteq \Sigma^*$ is said to be recognized by $M$ with error bound $\epsilon < \frac{1}{2}$ if and only if
\begin{itemize}
    \item every $x \in L$ is accepted by $M$ with probability at least $1 - \epsilon$ and
    \item every $x \notin L$ is rejected by $M$ with probability at least $1 - \epsilon$.
\end{itemize}
Then, it is also said that $L$ is recognized by $M$ with bounded error or bounded-error $M$ recognizes $L$.

Both bounded-error PFAs and QFAs recognize all and only regular languages \cite{Rab63,KW97,LQZLWM12}, and so, their computational power is equivalent to DFAs.

Bounded-error PostPFA can recognize some nonregular languages such as $\mathtt{EQUAL} = \{ a^n b a^n  \mid 
 n > 0 \}$ \cite{YS13A}. But they cannot recognize any nonstochastic language, as two-way PFAs with cutpoints recognize all and only stochastic languages \cite{Kan91}. Bounded-error PostQFA are more powerful than bounded-error PostPFA, i.e., $\mathtt{PAL} = \{ w \in \{a,b\}^* \mid w = w^R\}$ and $\mathtt{TWIN} = \{ wcw \mid w \in \{a,b\}^* \}$ are some witness languages \cite{FK94,AW02,YS10B}. Remark that any bounded-error postselecting algorithm is easily modified to recognize with cutpoint $\frac{1}{2}$. Thus, bounded-error PostQFAs cannot recognize any nonstochastic language \cite{YS13A}.

\section{Merlin-Arthur automata}
\label{sec:MA-def}

For a given input $x \in \Sigma^*$, the verifier (Arthur) receives and reads a finite length certificate, say $c \in \Gamma^*$, from the prover (Merlin) before reading the input: the automaton verifier reads 
\[
    (c,x) = \lmoon ~ c_1 ~ c_2 ~ \cdots ~ c_{|c|} ~ \rmoon ~ \cent ~ x_1 ~ x_2 ~ \cdots ~ x_{|x|} ~ \dollar
\]
from left to right and symbol by symbol and then make its decision on $(c,x)$.

We denote Merlin-Arthur DFA, PFA, and QFA as MA-DFA, MA-PFA, and MA-QFA, respectively. MA-PostPFA and MA-PostQFA are the postselecting versions of MA-PFA and MA-QFA, respectively. Their definitions are extended with the certificate alphabet $\Gamma$, and then the set of transition operators is extended for each symbol $\sigma \in \tildegamma$ that are not in $\tildesigma$.

\subsection{Language verification}

A language $L \subseteq \Sigma^*$ is verified by a MA-DFA $M$ if and only if for every input $x \in \Sigma^*$:
\begin{itemize}
    \item when $x \in L$, there is a certificate $c \in \Gamma^* $ such that $M$ accepts $(c,x)$, and,
    \item when $x \notin L$, for any certificate $ c \in \Gamma^* $, $M$ rejects $(c,x)$. 
\end{itemize}

The following definitions are for MA-PFAs and MA-QFAs.

A language $L \subseteq \Sigma^*$ is verified by a MA automaton $M$ with cutpoint $\lambda \in [0,1) $ if and only if for every input $x \in \Sigma^*$:
\begin{itemize}
    \item when $x \in L$, there is a certificate $c \in \Gamma^* $ such that $M$ accepts $(c,x)$  with probability greater than $\lambda$, and,
    \item when $x \notin L$, for any certificate $ c \in \Gamma^* $, $M$ accepts $(c,x)$ with probability at most $\lambda$. 
\end{itemize}

A language $L \subseteq \Sigma^*$ is verified by a MA automaton $M$ with error bound $\epsilon < \frac{1}{2}$ if and only if for every input $x \in \Sigma^*$:
\begin{itemize}
    \item when $x \in L$, there is a certificate $c \in \Gamma^* $ such that $M$ accepts $(c,x)$  with probability at least $1-\epsilon$, and,
    \item when $x \notin L$, for any certificate $ c \in \Gamma^* $, $M$ rejects $(c,x)$ with probability at least $1-\epsilon$. 
\end{itemize}
Then, it is also said that $L$ is verified by $M$ with bounded-error or bounded-error $M$ verifies $L$.

\subsection{``Weak'' Merlin-Arthur automata}

In this paper, we focus on only finite computation. That is, the certificates by Merlin can be arbitrarily long but they are still finite. In other words, we omit the infinite certificates by definition.

One may allow infinite certificates. Then, it can be the case that the verifier may never read the input, and so, no decision is made. Thus, we may not achieve high rejecting probabilities for the non-members, but at least the accepting probability will be zero.

A proof system is called ``weak'' \cite{DS92} if we replace the bounded-error language verification condition for the non-members as follow:
\begin{itemize}
    \item when $x \notin L$, for any certificate $c$ (including the infinite ones), $M$ accepts $(c,x) $ with probability at most $\epsilon$.
\end{itemize}
We remark that all results in this paper are obtained for ``weak'' Merlin-Arthur automata, if the infinite certificates are allowed.

\section{Verification with sublinear size certificates}
\label{sec:sublinear}

We start with constant size certificates.
MA-DFAs are always in one state during their computation. Therefore, after reading a certificate, all information about this certificate is stored as one of the states. Thus, instead of long certificates, Merlin simply says from which state the verifier should start its computation. If the MA-DFA verifier has $m$ states, then unary certificates with length less than $m$ and binary certificates with length at most $ \log_2 m$ are sufficient.

Let $L$ be a language verified by an $m$-state MA-DFA $M$ with the set of states $S$. For every $x \in L$, there is a certificate $c_x$ such that $M$ accepts $(c_x,x)$. After reading $c_x$, let $s(x) \in S$ be the state $M$ is in. That is, $x$ is accepted by $M$ when starting in state $s(x)$. In this way, we can pair every $x \in L$ with a state $s(x)$. By considering all member strings of $L$, we form a set of states as
\[ 
    S_I = \left\{ s(x) \mid x \in L \mbox{ and $M$ accepts $x$ when starting in $s(x)$} \right\}.
\]
Let $k$ be the size of $S_I$.

Suppose that for any given input $x \in \Sigma^*$, we execute $k$ copies of $M$ on $x$ and each copy starts in a different state in $S_I$. Then, the decision of ``acceptance'' is made if one copy accepts $x$. If none of copies accepts, then the decision of ``rejecting'' is given. Indeed, this is already a known model defined in the literature: $k$-entry DFA \cite{HolzerSY01}. If $k=n$, then, it is called multiple-entry finite automaton (MEFA) \cite{GillK74}. 

\begin{theorem}
    If a language $L$ is verified by an $m$-state MA-DFA, then $L$ is recognized by some $k$-entry DFAs with $m$ states, where $k \leq m$.
\end{theorem}

The other way simulation can also be obtained easily.

\begin{theorem}
    If a language $L$ is recognized by a $k$-entry DFA with $m$ states, say $M$, then $L $ is verified by an $(m+1)$-state MA-DFA. 
\end{theorem}
\begin{proof}
    The verifier MA-DFA checks whether the initial state (after reading the certificate) is one of the specified initial state of $M$. If so, it simulates $M$ on the input. If not, $M$ enters a non-accepting sink state (the additional state) to make sure to reject the given input. 
    
    For any $x \in L$, Merlin provides one of the correct initial states that MA-DFA starts in. For the case of $x \notin L$, the verifier either starts in one of the $k$ states or switches to the additional sink states. In either case, the input is rejected.
\end{proof}

Multiple-entry DFAs can be exponentially more succinct than DFAs on binary languages and this bound is tight \cite{GalilS76,VelosoG79}. On unary (and binary) languages, $k$-entry DFAs can be polynomially more succinct than DFAs \cite{HolzerSY01,Polak05}, and these bounds are also tight. Thus, MA-DFAs have the same state efficiency over the DFAs.

\begin{corollary}
    MA-DFAs can recognize all and only regular languages. On the other hand, they can be exponentially state-efficient over the binary languages and polynomially state-efficient over the unary languages.
\end{corollary}


We continue with MA-PFAs. Here is a unary nonstochastic language \cite{Tur81}: 
\[
    \usquare = \{ a^{i^2} \mid i \ \geq 0  \}.
\]

\begin{theorem}
    \label{thm:mapfa-usquare}
    Language $\usquare$ is verified by a MA-PFA $M$ with cutpoint $\frac{1}{2}$ such that for any member string with length $n$, a unary certificate with length $\sqrt{n}$ is used.
\end{theorem}
\begin{proof}
    The empty string is handled deterministically. For any given non-empty input $x = a^n$, the verifier expects to have a certificate $a^i$, satisfying $i^2 = n$ for the member strings of $\usquare$.

    Below is a linear operator to encode $i^2$ when started in $(1~~1~~0)^T$ and applied $i$ times:
    \begin{equation}
        \label{eq:pfa-i2}
        \underbrace{\myvector{1 \\ 2i+1 \\ i^2}}_{final} = \underbrace{\mymatrix{ccc}{1 & 0 & 0 \\ 2 & 1 & 0 \\ 0 & 1 & 1 }^i}_{operator} \underbrace{\myvector{1 \\ 1 \\ 0}}_{start}.
    \end{equation}
    Below is a linear operator to encode $n$ when started in $(1~~0)^T$ and applied $n$ times:
    \begin{equation}
        \label{eq:pfa-n}        
        \underbrace{\myvector{1 \\ n}}_{final} = \underbrace{\mymatrix{cc}{1 & 0 \\ 1 & 1}^n}_{operator} \underbrace{\myvector{1 \\ 0}}_{start}.
    \end{equation}
    Both of these linear operators are embedded into the transition matrices of $M$ with some normalization factors. Here an additional state is used to ensure that the transition matrices are stochastic and the normalization factor, say $l$, is the same for each matrix. The technical details are given below.
    
    After reading $(0^i,0^n)$, the operators in Equations \ref{eq:pfa-i2} and \ref{eq:pfa-n} are applied $i$ and $n$ times, respectively, and
    $M$ sets its final state vector as
    \[
        v_f  =  l^{i+n+4}
        \myvector{i^2 \\ n \\ a = (l^{i+n+4}-i^2-n)/2 \\ r = (l^{i+n+4}-i^2-n)/2 \\ 0 \\ \vdots \\ 0 },
    \]
    where 
    \begin{itemize}
        \item $l^{i+n+4}$ is the normalization factor after the multiplication of $(i+n+4)$ transition matrices, and,
        \item after setting the first two entries with $i^2$ and $n$, the remaining probabilities are split into equal parts and set as the values of states $s_3$ and $s_4$.
    \end{itemize}
    The accepting states are $s_1$ and $s_3$.

    If $x \in L$, then with the help of a valid certificate, the accepting and rejecting probabilities are the same, and so, $Acc_M(x) = 1/2 $. 
    
    If $x \notin L$, $i^2$ and $n$ are always different for any certificate. Thus, we have $ Acc_M(x) \neq 1/2 $.

    When using only rational numbers, the above set-up can be converted into the acceptance with cutpoint $1/2$ by using a standard technique given in  \cite{Tur71}. 
    
    Now we provide the missing details. The verifier $M$ has 11 states, and we split them into three groups: 
    \begin{itemize}
        \item The states $\{ s_1,\ldots,s_5\}$ are the main states, with which we do our probabilistic computation.
        \item The states $\{s_6,\ldots,s_{10}\} $ are used to deterministically check if the input is empty string.
        \item The state $s_{11}$ is used for making the transition matrices stochastic.
    \end{itemize}
    For a given $a^n$, $M$ reads $(i+n+4)$ symbols:
    \[
    \lmoon a^i \rmoon \cent a^n \dollar .
    \]
    The initial state vector is 
    \[ 
        v_0 = \myvector{1 \\ 0 \\ \vdots \\ 0} .
    \]
    
    The transition matrix for $\lmoon$ is
    \[
        A_{\lmoon} = 
        \dfrac{1}{4} \mymatrix{ccccc|ccccc|c}{
        1 & 0 & 0 & 0 & 0 & 0 & 0 & 0 & 0 & 0 & 0  \\
        1 & 4 & 0 & 0 & 0 & 0 & 0 & 0 & 0 & 0 & 0 \\
        0 & 0 & 4 & 0 & 0 & 0 & 0 & 0 & 0 & 0 & 0 \\
        0 & 0 & 0 & 4 & 0 & 0 & 0 & 0 & 0 & 0 & 0 \\ 
        0 & 0 & 0 & 0 & 4 &  0 & 0 & 0 & 0 & 0 & 0 \\
        \hline
        0 & 0 & 0 & 0 & 0 & 4 & 0 & 0 & 0 & 0 & 0 \\
        0 & 0 & 0 & 0 & 0 & 0 & 4 & 0 & 0 & 0 & 0 \\
        0 & 0 & 0 & 0 & 0 & 0 & 0 & 4 & 0 & 0 & 0 \\
        0 & 0 & 0 & 0 & 0 & 0 & 0 & 0 & 4 & 0 & 0 \\
        0 & 0 & 0 & 0 & 0 & 0 & 0 & 0 & 0 & 4 & 0 \\
        \hline
        2 & 0 & 0 & 0 & 0 & 0 & 0 & 0 & 0 & 0 & 4 \\
        } .
    \]
    After reading $\lmoon$, the state vector is 
    \[
        v_1 = \dfrac{1}{4}
        \myvector{1 \\ 1 \\ 0 \\ \vdots \\ 0 \\ 2 } .
    \]
    
    The transition matrix for input symbol $a$ is
    \[
        A_a = 
        \dfrac{1}{4} \mymatrix{ccccc|ccccc|c}{
        1 & 0 & 0 & 0 & 0 & 1 & 0 & 0 & 0 & 0 & 0  \\
        2 & 1 & 0 & 0 & 0 & 2 & 1 & 0 & 0 & 0 & 0 \\
        0 & 1 & 1 & 0 & 0 & 0 & 1 & 1 & 0 & 0 & 0 \\
        1 & 0 & 0 & 1 & 0 & 1 & 0 & 0 & 1 & 0 & 0 \\ 
        0 & 0 & 0 & 0 & 1 &  0 & 0 & 0 & 0 & 1 & 0 \\
        \hline
        0 & 0 & 0 & 0 & 0 & 0 & 0 & 0 & 0 & 0 & 0 \\
        0 & 0 & 0 & 0 & 0 & 0 & 0 & 0 & 0 & 0 & 0 \\
        0 & 0 & 0 & 0 & 0 & 0 & 0 & 0 & 0 & 0 & 0 \\
        0 & 0 & 0 & 0 & 0 & 0 & 0 & 0 & 0 & 0 & 0 \\
        0 & 0 & 0 & 0 & 0 & 0 & 0 & 0 & 0 & 0 & 0 \\
        \hline
        0 & 2 & 3 & 3 & 3 & 0 & 2 & 3 & 3 & 3 & 4 \\
        } .
    \]
    Here the linear operator in Equation~\ref{eq:pfa-i2} is embedded into the 1st, 2nd, and 3rd rows and columns, and, the linear operator in Equation~\ref{eq:pfa-n} is embedded into the 1st and 4th rows and columns. To use encoding, we set the initial conditions accordingly. The first three entries of $v_1$ (by omitting the normalization factor) is 
    \[
        \myvector{1 \\ 1 \\ 0}.
    \]
    Thus, if we apply $A_a$, we use the encoding given in Equation~\ref{eq:pfa-i2}. The fourth entry of $v_1$ is zero. So, the encoding given in Equation~\ref{eq:pfa-n} is not activated at the moment.
    
    After reading the proof $a^i$, as shown in Equation~\ref{eq:pfa-i2}, the state vector is
    \[
        v_{i+1} = 
        \mypar{\dfrac{1}{4}}^{i+1}
        \myvector{1 \\ 2i+1 \\ i^2 \\ 0 \\ 0 \\ \hline 0 \\ \vdots \\ 0 \\ \hline T_{i+1} },
    \]
    where $T_{i+1}$ is some value to make the vector stochastic.
    
    The transition matrix for $\rmoon$ is 
    \[
        A_{\rmoon} = 
        \dfrac{1}{4} \mymatrix{ccccc|ccccc|c}{
        1 & 0 & 0 & 0 & 0 & 0 & 0 & 0 & 0 & 0 & 0  \\
        0 & 0 & 0 & 0 & 0 & 0 & 0 & 0 & 0 & 0 & 0 \\
        0 & 0 & 0 & 0 & 0 & 0 & 0 & 0 & 0 & 0 & 0 \\
        1 & 0 & 0 & 0 & 0 & 0 & 0 & 0 & 0 & 0 & 0 \\ 
        0 & 0 & 1 & 0 & 0 &  0 & 0 & 0 & 0 & 0 & 0 \\
        \hline
        0 & 0 & 0 & 0 & 0 & 0 & 0 & 0 & 0 & 0 & 0 \\
        0 & 0 & 0 & 0 & 0 & 0 & 0 & 0 & 0 & 0 & 0 \\
        0 & 0 & 0 & 0 & 0 & 0 & 0 & 0 & 0 & 0 & 0 \\
        0 & 0 & 0 & 0 & 0 & 0 & 0 & 0 & 0 & 0 & 0 \\
        0 & 0 & 0 & 0 & 0 & 0 & 0 & 0 & 0 & 0 & 0 \\
        \hline
        2 & 4 & 3 & 4 & 4 & 4 & 4 & 4 & 4 & 4 & 4 \\
        } .
    \]
    After reading $\rmoon$, the state vector is
    \[
        v_{i+2} = \mypar{\dfrac{1}{4}}^{i+2}
        \myvector{1 \\ 0 \\ 0 \\ 1 \\ i^2 \\ \hline 0 \\ \vdots \\ 0 \\ \hline T_{i+2} },
    \]
    where $T_{i+2}$ is some value to make the vector stochastic.
    Here we set the fourth entry to 1 for activating the encoding in Equation~\ref{eq:pfa-n}. We store $i^2$ on the fifth entry to be used at the end of computation. We set the second and third entries to zero to deactivate the encoding in Equation~\ref{eq:pfa-i2}.
    
    The transition matrix for $\cent$ is 
    \[  
        A_{\cent} = 
        \dfrac{1}{4} \mymatrix{ccccc|ccccc|c}{
        0 & 0 & 0 & 0 & 0 & 0 & 0 & 0 & 0 & 0 & 0  \\
        0 & 0 & 0 & 0 & 0 & 0 & 0 & 0 & 0 & 0 & 0 \\
        0 & 0 & 0 & 0 & 0 & 0 & 0 & 0 & 0 & 0 & 0 \\
        0 & 0 & 0 & 0 & 0 & 0 & 0 & 0 & 0 & 0 & 0 \\ 
        0 & 0 & 0 & 0 & 0 &  0 & 0 & 0 & 0 & 0 & 0 \\
        \hline
        1 & 0 & 0 & 0 & 0 & 0 & 0 & 0 & 0 & 0 & 0 \\
        0 & 1 & 0 & 0 & 0 & 0 & 0 & 0 & 0 & 0 & 0 \\
        0 & 0 & 1 & 0 & 0 & 0 & 0 & 0 & 0 & 0 & 0 \\
        0 & 0 & 0 & 1 & 0 & 0 & 0 & 0 & 0 & 0 & 0 \\
        0 & 0 & 0 & 0 & 1 & 0 & 0 & 0 & 0 & 0 & 0 \\
        \hline
        3 & 3 & 3 & 3 & 3 & 4 & 4 & 4 & 4 & 4 & 4 \\
        } .
    \]
    After reading $\cent$, the state vector is
    \[
        v_{i+3} = \mypar{\dfrac{1}{4}}^{i+3}
        \myvector{0 \\ \vdots \\ 0 \\ \hline 1 \\ 0 \\ 0 \\ 1 \\ i^2 \\ \hline T_{i+3} },
    \]
    where $T_{i+3}$ is some value to make the vector stochastic. Here we transfer our ``useful'' computation so far in the second five entries.
    After $\cent$, if we read a symbol $a$, our computation so far will move back to the first five entries. Otherwise, if we immediately read $\dollar$, we make sure that the input (empty string) is rejected with probability 1 (deterministically).
    
    Suppose that $a^n$ is non-empty. After reading $a^n$, as shown in Equation~\ref{eq:pfa-n}, the state vector is 
    \[
        v_{i+n+3} = \mypar{\dfrac{1}{4}}^{i+n+3}
        \myvector{1 \\ 0 \\ 0 \\ n \\ i^2 \\ \hline   0 \\ \vdots \\ 0 \\ \hline T_{i+n+3} },
    \]
     where $T_{i+n+3}$ is some value to make the vector stochastic.
     
    The transition matrix for $\dollar$ is 
    \[
        A_{\dollar} = 
        \dfrac{1}{4} \mymatrix{ccccc|ccccc|c}{
        0 & 0 & 0 & 0 & 1 & 0 & 0 & 0 & 0 & 0 & 0  \\
        0 & 0 & 0 & 1 & 0 & 0 & 0 & 0 & 0 & 0 & 0 \\
        0 & 0 & 0 & 0 & 0 & 0 & 0 & 0 & 0 & 0 & 2 \\
        0 & 0 & 0 & 0 & 0 & 0 & 0 & 0 & 0 & 0 & 2 \\ 
        0 & 0 & 0 & 0 & 0 &  0 & 0 & 0 & 0 & 0 & 0 \\
        \hline
        0 & 0 & 0 & 0 & 0 & 1 & 0 & 0 & 0 & 0 & 0 \\
        0 & 0 & 0 & 0 & 0 & 0 & 1 & 0 & 0 & 0 & 0 \\
        0 & 0 & 0 & 0 & 0 & 0 & 0 & 1 & 0 & 0 & 0 \\
        0 & 0 & 0 & 0 & 0 & 0 & 0 & 0 & 1 & 0 & 0 \\
        0 & 0 & 0 & 0 & 0 & 0 & 0 & 0 & 0 & 1 & 0 \\
        \hline
        4& 4 & 4 & 3 & 3 & 3 & 3 & 3 & 3 & 3 & 0 \\
        } .
    \]
    If $a^n$ is empty string, the state vector is
    \[
        v_{i+4} = v_f = 
        \mypar{\dfrac{1}{4}}^{i+4}
        \myvector{0 \\ 0 \\ 2T^{i+3} \\ 2T^{i+3} \\ 0 \\ \hline 1 \\ 0 \\ 0 \\ 1 \\ i^2 \\ \hline 0}.
    \]
    It is clear that $Acc_M(\varepsilon) \neq 1/2$.
    If $a^n$ is a  non-empty string, the state vector is
    \[
        v_{i+n+4} = v_f = \mypar{\dfrac{1}{4}}^{i+n+4}
        \myvector{i^2 \\ n \\ 2T^{i+n+3} \\ 2T^{i+n+3} \\ 0 \\ \hline 0 \\ \vdots \\ 0 \\ \hline 0}.
    \]
    Thus, if $a^n$ is in the language, then $i^2 = n$ and so $Acc_M(x) = 1/2$. Otherwise, $i^2 \neq n$ and so $Acc_M(x) \neq 1/2$.
\end{proof}

In the proof above, we may also encode $i^2$ as 
\[
    \underbrace{\myvector{1 \\ i \\ i \\ i^2}}_{final} = \underbrace{\myvector{1 \\ i } \otimes \myvector{1 \\ i}}_{final}  = \underbrace{\left( \mymatrix{cc}{1 & 0 \\ 1 & 1} \otimes \mymatrix{cc}{1 & 0 \\ 1 & 1} \right)^i}_{operator} \cdot \underbrace{\myvector{1 \\ 0 } \otimes \myvector{1 \\0 }}_{start}.
\]
In the same way, we can encode $i^k$ for some $k>2$ by using more tensoring. Any unary polynomial language $\mathtt{UPOLY(k)} = \{ a^{i^k} \mid i \geq 0 \}$ is also nonstochastic \cite{Tur81}, and, we can modify the above proof for any of them in a straightforward way.

\begin{corollary}
    For any $k>2$, language $\mathtt{UPOLY(k)}$ is verified by a MA-PFA $M$ with cutpoint $\frac{1}{2}$ such that for any member string with length $n$, a unary certificate with length $\sqrt[k]{n}$ is used.
\end{corollary}

We do not know whether nonregular unary language $\mathtt{UPOWER} = \{ a^{2^i} \mid i \geq 0 \}$ is nonstochastic. We use only logarithmic size certificates for $\mathtt{UPOWER}$ based on the following linear operator:
\[
    \underbrace{\myvector{1 \\ 2^i}}_{final} = \underbrace{\mymatrix{cc}{1 & 0 \\ 0 & 2}^i}_{operator} \underbrace{\myvector{1 \\ 1}}_{start} \cdot
\]

\begin{corollary}
    Language $\mathtt{UPOWER}$ is verified by a MA-PFA $M$ with cutpoint $\frac{1}{2}$ such that for any member string with length $n$, a unary certificate with length $\log n$ is used.
\end{corollary}

The above results presented for MA-PFAs can be obtained for bounded-error MA-PostQFAs. The proof is almost the same. We embed the linear operators inside the quantum operators by using normalization factors and filling the rest to ensure the matrices form valid quantum operators. One main difference is that, due to negative transitions, QFAs can subtract the values from each other, and, in this way, we can use the value of $i^2-n$, which is zero if and only if $n$ and $i^2$ are equal. Overall, this is one of the well-known programming techniques for QFAs. We refer the reader to \cite{YS10A,YS11A,Yak13A,Yak13C,Yak16A,SY17} for the details and several examples. For the sake of completeness, we provide these details in the next proof.

\begin{theorem}
    Language $\usquare \cup \{\varepsilon\}$ is verified by a bounded-error MA-PostQFA $M$ such that for any member string with length $n$, a unary certificate with length $\sqrt{n}$ is used.
\end{theorem}
\begin{proof}
    We include $\varepsilon$ as an accepting string to avoid a deterministic check and so to use more states, which makes the technical details more complicated.
    
    QFAs can have both negative and positive transition values. Thus, for a given input $x = a^n$, after reading the certificate $a^i$ and the input, the (unnormalized) final state of the verifier is set to
    \[
        l^{n+i+4} \myvector{1 \\ t(n - i^2) \\ 0 \\ \vdots \\ 0},
    \]
    where \begin{itemize}
        \item $l^{n+i+4}$ is the cumulative normalization factor coming from each individual operator for every symbol of $ \lmoon a^i \rmoon \mspace{3mu} \cent \mspace{3mu}  a^n \mspace{3mu} \dollar $, and,
        \item $t$ is a positive integer to make error bound arbitrarily close to 0.
    \end{itemize}
    All the other transitions are terminated in non-postselecting state(s). When using superoperators, we keep our ``useful'' computation by one operation element, and the rest of operation element(s) move the computation to the non-postselecting state(s).
    
    The only accepting state is $q_1$ and the only rejecting state is $q_2$. Thus, if $x \in L$, $t(n-i^2) = 0 $, and so, the input is accepted with probability 1. If $x \notin L$, then, (as $|n-i^2| \geq 1$, the input is rejected with probability at least
    \[
    Rej_M(x) = \dfrac{t^2}{t^2+1},
    \]
     which can be arbitrarily close to 1 by picking big $t$ values.

     Now we provide the technical details. The verifier $M$ has 5 main states and 5 auxiliary states. For each symbol, we define a superoperator with two or three operation elements. The first operation elements are for our computation. The others are auxiliary to make the overall operator is a superoperator. For simplicity, we set $t=2$.

     We start with the superoperator $\mathcal{E}_a = \{E_{a,1}, E_{a,2},E_{a,3}\}$:
     \[
        E_{a,1} = \dfrac{1}{3} \mymatrix{rrrrr|rrrrr}{
        1 & 0 &  0 &  0 &  0 &  0 &  0 &  0 &  0 &  0 \\
        2 & 1 &  0 &  0 &  0 &  0 &  0 &  0 &  0 &  0 \\
        0 & 1 &  1 &  0 &  0 &  0 &  0 &  0 &  0 &  0 \\
        1 & 0 &  0 &  1 &  0 &  0 &  0 &  0 &  0 &  0 \\
        0 & 0 &  0 &  0 &  1 &  0 &  0 &  0 &  0 &  0 \\
        \hline
        1 & -2 &  0 &  -1 &  0 &  0 &  0 &  0 &  0 &  0 \\
        0 & 1 &  -1 &  0 &  0 &  0 &  0 &  0 &  0 &  0 \\
        0 & 0 &  0 &  0 &  0 &  0 &  0 &  0 &  0 &  0 \\
        0 & 0 &  0 &  0 &  0 &  0 &  0 &  0 &  0 &  0 \\
        0 & 0 &  0 &  0 &  0 &  0 &  0 &  0 &  0 &  0
        }
     \]
     The top left corner is for encoding as we do in the proof of Theorem~\ref{thm:mapfa-usquare}. The bottom left part is filled with some values to ensure that the columns are pairwise orthogonal. We define $E_{a,2} $ and $ E_{a,3}$ to make $ \mathcal{E}_a $ a valid quantum operator.
     \[
        E_{a,2} = \dfrac{1}{3} \mymatrix{rrrrr|rrrrr}{
        0 & 0 &  0 &  0 &  0 &  0 &  0 &  0 &  0 &  0 \\
        0 & 0 &  0 &  0 &  0 &  0 &  0 &  0 &  0 &  0 \\
        0 & 0 &  0 &  0 &  0 &  0 &  0 &  0 &  0 &  0 \\
        0 & 0 &  0 &  0 &  0 &  0 &  0 &  0 &  0 &  0 \\
        0 & 0 &  0 &  0 &  0 &  0 &  0 &  0 &  0 &  0 \\
        \hline
        \sqrt{2} & 0 &  0 &  0 &  0 &  0 &  0 &  0 &  0 &  0 \\
        0 & \sqrt{2} &  0 &  0 &  0 &  0 &  0 &  0 &  0 &  0 \\
        0 & 0 &  \sqrt{7} &  0 &  0 &  0 &  0 &  0 &  0 &  0 \\
        0 & 0 &  0 &  \sqrt{7} &  0 &  0 &  0 &  0 &  0 &  0 \\
        0 & 0 &  0 &  0 &  \sqrt{8} &  0 &  0 &  0 &  0 &  0
        }
     \]
     \[
        E_{a,3} = \dfrac{1}{3} \mymatrix{rrrrr|rrrrr}{
        0 & 0 &  0 &  0 &  0 &  0 &  0 &  0 &  0 &  0 \\
        0 & 0 &  0 &  0 &  0 &  0 &  0 &  0 &  0 &  0 \\
        0 & 0 &  0 &  0 &  0 &  0 &  0 &  0 &  0 &  0 \\
        0 & 0 &  0 &  0 &  0 &  0 &  0 &  0 &  0 &  0 \\
        0 & 0 &  0 &  0 &  0 &  0 &  0 &  0 &  0 &  0 \\
        \hline
        0 & 0 &  0 &  0 &  0 &  3 &  0 &  0 &  0 &  0 \\
        0 & 0 &  0 &  0 &  0 &  0 &  3 &  0 &  0 &  0 \\
        0 & 0 & 0  &  0 &  0 &  0 &  0 &  3 &  0 &  0 \\
        0 & 0 &  0 &  0 &  0 &  0 &  0 &  0 &  3 &  0 \\
        0 & 0 &  0 &  0 &  0 &  0 &  0 &  0 &  0 &  3
        }
     \]
    Now it is easy to check that the columns of 
    \[
    \begin{array}{|c|} 
    \hline  E_{a,1}  \\ \hline E_{a,2} \\ \hline  E_{a,3} \\ \hline
    \end{array}
    \]
    form an orthonormal set, i.e., the columns are pairwise orthogonal and their lengths are 1. Thus, $\mathcal{E}_a$ is a valid superoperator, and it also ensures that the auxiliary states are sink (there are no transitions from the auxiliary states to the main states). 

    In the same way, we define the superoperators for the other symbols. The superoperator for symbol $\lmoon$ has two operation elements:
    \[
        E_{\lmoon,1} = \dfrac{1}{3} 
        \mymatrix{rrrrr|rrrrr}{
        1 & 0 &  0 &  0 &  0 &  0 &  0 &  0 &  0 &  0 \\
        1 & 0 &  0 &  0 &  0 &  0 &  0 &  0 &  0 &  0 \\
        0 & 0 &  0 &  0 &  0 &  0 &  0 &  0 &  0 &  0 \\
        0 & 0 &  0 &  0 &  0 &  0 &  0 &  0 &  0 &  0 \\
        0 & 0 &  0 &  0 &  0 &  0 &  0 &  0 &  0 &  0 \\
        \hline
        0 & 0 &  0 &  0 &  0 &  3 &  0 &  0 &  0 &  0 \\
        0 & 0 &  0 &  0 &  0 &  0 &  3 &  0 &  0 &  0 \\
        0 & 0 & 0  &  0 &  0 &  0 &  0 &  3 &  0 &  0 \\
        0 & 0 &  0 &  0 &  0 &  0 &  0 &  0 &  3 &  0 \\
        0 & 0 &  0 &  0 &  0 &  0 &  0 &  0 &  0 &  3
        }
    \]
    \[
        E_{\lmoon,2} = \dfrac{1}{3} 
        \mymatrix{rrrrr|rrrrr}{
        0 & 0 &  0 &  0 &  0 &  0 &  0 &  0 &  0 &  0 \\
        0 & 0 &  0 &  0 &  0 &  0 &  0 &  0 &  0 &  0 \\
        0 & 0 &  0 &  0 &  0 &  0 &  0 &  0 &  0 &  0 \\
        0 & 0 &  0 &  0 &  0 &  0 &  0 &  0 &  0 &  0 \\
        0 & 0 &  0 &  0 &  0 &  0 &  0 &  0 &  0 &  0 \\
        \hline
        \sqrt{7} & 0 &  0 &  0 &  0 &  0 &  0 &  0 &  0 &  0 \\
        0 & 3 & 0  &  0 &  0 &  0 &  0 &  0 &  0 &  0 \\
        0 & 0 & 3  &  0 &  0 &  0 &  0 &  0 &  0 &  0 \\
        0 & 0 &  0 &  3 &  0 &  0 &  0 &  0 &  0 &  0 \\
        0 & 0 &  0 &  0 &  3 &  0 &  0 &  0 &  0 &  0
        }
    \]
    The superoperator for symbol $\rmoon$ has two operation elements:
    \[
        E_{\rmoon,1} = \dfrac{1}{3} 
        \mymatrix{rrrrr|rrrrr}{
        1 & 0 &  0 &  0 &  0 &  0 &  0 &  0 &  0 &  0 \\
        0 & 0 &  0 &  0 &  0 &  0 &  0 &  0 &  0 &  0 \\
        0 & 0 &  0 &  0 &  0 &  0 &  0 &  0 &  0 &  0 \\
        1 & 0 &  0 &  0 &  0 &  0 &  0 &  0 &  0 &  0 \\
        0 & 0 &  1 &  0 &  0 &  0 &  0 &  0 &  0 &  0 \\
        \hline
        0 & 0 &  0 &  0 &  0 &  3 &  0 &  0 &  0 &  0 \\
        0 & 0 &  0 &  0 &  0 &  0 &  3 &  0 &  0 &  0 \\
        0 & 0 & 0  &  0 &  0 &  0 &  0 &  3 &  0 &  0 \\
        0 & 0 &  0 &  0 &  0 &  0 &  0 &  0 &  3 &  0 \\
        0 & 0 &  0 &  0 &  0 &  0 &  0 &  0 &  0 &  3
        }
    \]
    \[
        E_{\rmoon,2} = \dfrac{1}{3} 
        \mymatrix{rrrrr|rrrrr}{
        0 & 0 &  0 &  0 &  0 &  0 &  0 &  0 &  0 &  0 \\
        0 & 0 &  0 &  0 &  0 &  0 &  0 &  0 &  0 &  0 \\
        0 & 0 &  0 &  0 &  0 &  0 &  0 &  0 &  0 &  0 \\
        0 & 0 &  0 &  0 &  0 &  0 &  0 &  0 &  0 &  0 \\
        0 & 0 &  0 &  0 &  0 &  0 &  0 &  0 &  0 &  0 \\
        \hline
        \sqrt{7} & 0 &  0 &  0 &  0 &  0 &  0 &  0 &  0 &  0 \\
        0 & 3 & 0  &  0 &  0 &  0 &  0 &  0 &  0 &  0 \\
        0 & 0 & \sqrt{8}  &  0 &  0 &  0 &  0 &  0 &  0 &  0 \\
        0 & 0 &  0 &  3 &  0 &  0 &  0 &  0 &  0 &  0 \\
        0 & 0 &  0 &  0 &  3 &  0 &  0 &  0 &  0 &  0
        }
    \]
    The superoperator for symbol $\cent$ has two operation elements:
    \[
        E_{\cent,1} = \dfrac{1}{3}
        \mymatrix{c|c}{ I & 0 \\ \hline 0 & 3I }
    \]
    \[
        E_{\cent,2} = \dfrac{1}{3}
        \mymatrix{c|c}{ 0 & 0 \\ \hline \sqrt{8}I & 0 }
    \]
    As we do not make a deterministic check for empty string, we apply a fraction of the identity operator on the main part.
    Now we define the superoperator with two operation elements for symbol $\dollar$. Note that the encoding of $i^2$ and $n$ are on the fifth and fourth entries respectively.
    \[
        E_{\dollar,1} = \dfrac{1}{3} 
        \mymatrix{rrrrr|rrrrr}{
        1 & 0 &  0 &  0 &  0 &  0 &  0 &  0 &  0 &  0 \\
        0 & 0 &  0 &  0 &  0 &  0 &  0 &  0 &  0 &  0 \\
        0 & 0 &  0 &  0 &  0 &  0 &  0 &  0 &  0 &  0 \\
        0 & 2 &  0 &  0 &  0 &  0 &  0 &  0 &  0 &  0 \\
        0 & -2 & 0  &  0 &  0 &  0 &  0 &  0 &  0 &  0 \\
        \hline
        0 & 0 &  0 &  0 &  0 &  3 &  0 &  0 &  0 &  0 \\
        0 & 0 &  0 &  0 &  0 &  0 &  3 &  0 &  0 &  0 \\
        0 & 0 & 0  &  0 &  0 &  0 &  0 &  3 &  0 &  0 \\
        0 & 0 &  0 &  0 &  0 &  0 &  0 &  0 &  3 &  0 \\
        0 & 0 &  0 &  0 &  0 &  0 &  0 &  0 &  0 &  3
        }
    \]
    \[
        E_{\dollar,2} = \dfrac{1}{3} 
        \mymatrix{rrrrr|rrrrr}{
        0 & 0 &  0 &  0 &  0 &  0 &  0 &  0 &  0 &  0 \\
        0 & 0 &  0 &  0 &  0 &  0 &  0 &  0 &  0 &  0 \\
        0 & 0 &  0 &  0 &  0 &  0 &  0 &  0 &  0 &  0 \\
        0 & 0 &  0 &  0 &  0 &  0 &  0 &  0 &  0 &  0 \\
        0 & 0 &  0 &  0 &  0 &  0 &  0 &  0 &  0 &  0 \\
        \hline
        \sqrt{8} & 0 &  0 &  0 &  0 &  0 &  0 &  0 &  0 &  0 \\
        0 & \sqrt{5} & 0  &  0 &  0 &  0 &  0 &  0 &  0 &  0 \\
        0 & 0 & 3  &  0 &  0 &  0 &  0 &  0 &  0 &  0 \\
        0 & 0 &  0 &  3 &  0 &  0 &  0 &  0 &  0 &  0 \\
        0 & 0 &  0 &  0 &  3 &  0 &  0 &  0 &  0 &  0
        }
    \]

    The first and second states are postselecting accepting and rejecting state, respectively. The other states are non-postselecting states. Thus, we do not take into account the values of auxiliary states. After reading $ \lmoon a^i \rmoon \cent a^n \dollar $, the main part of the unnormalized final state vector is
    \[
        \mypar{\dfrac{1}{3}}^{n+i+4} 
        \myvector{1 \\ 2(n-i^2) \\ 0 \\ 0 \\ 0}.
    \]
    
    If $a^n$ is in the language, we know that $n=i^2$ and so $2(n-i^2) = 0 $. Thus, the input is accepted with probability 1. Otherwise,  $(n-i^2)^2$ takes 1 as the minimum value, and so, the input is rejected with probability at least $\dfrac{2^2}{2^2+1} = \dfrac{4}{5}$.

    We remark that for different value of $t$, it is enough to modify $\mathcal{E}_{\dollar}$. Overall, we can also get rid of irrational transition values and use only the rational numbers by using more operation elements.
\end{proof}

\begin{corollary}
    For any $k \geq 2$, language $\mathtt{UPOLY(k)}$ is verified by a bounded-error MA-PostQFA such that for any member string with length $n$, a unary certificate with length $\sqrt[k]{n}$ is used.
\end{corollary}

\begin{corollary}
    Language $\mathtt{UPOWER}$ is verified by a bounded-error MA-PostQFA $M$ such that for any member string with length $n$, a unary certificate with length $\log n$ is used.
\end{corollary}

\section{Verification with linear size certificates}
\label{sec:linear}

Language $\usquare$ can also be verified by bounded-error MA-PostPFA but we obtain this result with linear size certificates.

Bounded-error two-way PFAs can make unary equality check as well as a sequence of several unary equality checks \cite{Fre81,CL89,YS10B} such as
\[
    \mathtt{UEQUAL} = \{ a^n b a^n \mid n \geq 0 \}
\]
or
\[
    \mathtt{UEQUAL\mbox{-}BLOCKS} = \{ a^{n_1}b a^{n_1} a^{n_2}b a^{n_2} \cdots a^{n_k}b a^{n_k} \mid k > 1  \}
\]
or 
\[
    \mathtt{UEQUALS} = \{ a^n b a^n b a^n b \cdots b a^n b \mid n \geq 0 \}.
\]
The same checks can also be made by bounded-error PostPFAs as presented in \cite{DY18A,DY18A-arXiv}, i.e., the input is read once but we use the power of  postselection. Therefore, we present our verification protocols by reducing them to a series of equality checks. We also refer the reader to \cite{DY18A,DY18A-arXiv} for similar verification protocols by PostPFA, which are part of one-way private-coin IPSs.

\begin{theorem}
    Language $\usquare$ is verified by a bounded-error MA-PostPFA $M$ using linear size binary certificates.
\end{theorem}
\begin{proof}
    Trivial checks are made deterministically.
    Let $x = a^n$ be a given input. If $x $ is a member string, then there exists $i$ satisfying that $n = i^2$, and the expected certificate is
    \[
        a^i \underbrace{b}_{1st} a^i \underbrace{b}_{2nd} \cdots a^i \underbrace{b}_{ith}.
    \]
    In general, $M$ reads the certificate $c_x$ and the input $x$ as
    \[
    \lmoon a^{t_1} b a^{t_2} b \cdots a^{t_k} b \rmoon \cent a^n \dollar
    \]
    for some $k \geq 1$
    Then, $M$ makes the following checks:
    \begin{itemize}
        \item $t_1 = t_2$ and $ t_3 = t_4 $ and so on.
        \item $t_2=t_3$ and $t_4=t_5$ and so on.
        \item $t_1 = k$ ($k$ is the number of $b$s).
        \item $ t_1+t_2+\cdots+t_k = n$.
    \end{itemize}
    For the member string $x$, the certificate $c_x$ is provided as expected, and so, the input is accepted with high probability. For the non-member string $x$, at least one check must fail, and so, the input is rejected with high probability.
\end{proof}

For the input $a^{i^3}$, we can use some certificates of the form
\[
    \underbrace{ \underbrace{a^i b a^i b \cdots a^i b}_{\#b = i} ~ c ~ \underbrace{a^i b a^i b \cdots a^i b}_{\#b = i} ~ c ~ \cdots ~c~ \underbrace{a^i b a^i b \cdots a^i b}_{\#b = i} ~c }_{\#c = i}.
\]
Thus, we have $i^3$ $a$s in total. For simplicity, we use separators $b$ and $c$, but, the whole certificate be binary by keeping linear size. We can generalize this certificate form for the input $a^{i^k}$.

\begin{theorem}
    Language $\mathtt{UPOLY(k)}$ are verified by bounded-error MA-PostPFAs using linear size binary certificates.
\end{theorem}

For the input $a^{2^i}$, we can use some certificates of the form
\[
    a^{2^0} b a^{2^1} b \cdots ba^{2^{i-1}},
\]
where the number of $a$s is $1+2+\cdots+2^{i-1} = 2^{i} - 1 $. This time, $b$s are only the separator, and we do not count their numbers. Remark that, for a given certificate,
\[
    a^{t_1}ba^{t_2}b \cdots b a^{t_k},
\]
we do the following equality checks: 
\begin{itemize}
    \item $t_1 = 1$, and,
    \item $t_j = \frac{t_{j+1}}{2}$ for each $j=1,\ldots,k-1$.
\end{itemize}
Here we deterministically check if each of $t_2,\ldots,t_k$ is even for $k>1$, and then, for each $\frac{t_{j+1}}{2}$ mentioned above, we count a pair of $a$s as a single $a$.

\begin{theorem}
    Language $\mathtt{UPOWER}$ is verified by a bounded-error MA-PostPFA using linear size binary certificates.
\end{theorem}

MA-PostQFAs using linear size binary certificates are very powerful. They can indeed verify every unary language with bounded error. In \cite{SY17}, it was shown that two-way QFAs (with classical head) can verify any unary languages with exponential expected time where the QFA verifier is a part of AM proof systems. We adopt the given proof in \cite{SY17} for our case.\footnote{It was also adopted for affine automata verifiers as a part of AM proof systems in \cite{KY21A}.} 

Let $L \subseteq \{a\}^*$ be a unary language. We can put all unary stings in lexicographical order as: 
\[
    \begin{array}{lcl}
       \Sigma^*[1]  & = & a^0  \\
       \Sigma^*[2]  & = & a^1  \\
       \Sigma^*[3]  & = & a^2  \\
       \vdots & \vdots & \vdots \\
       \Sigma^*[i+1]  & = & a^i  \\
       \vdots & \vdots & \vdots
    \end{array}
\]
We define a function $F_L$ as
\[
    F_L(i) = \left\lbrace
    \begin{array}{cc}
        0 & \mbox{if } \Sigma^*[i] \notin L \\
        1  & \mbox{if } \Sigma^*[i] \in L
    \end{array}
    \right . .
\]
Then, we define a real number to store all membership information of $L$ as
\[
    \alpha_L = \sum_{i=1}^\infty \dfrac{F_L(i)}{4^i} = \dfrac{F_L(1)}{4} +  \dfrac{F_L(2)}{4^2} +  \dfrac{F_L(3)}{4^3} + \cdots .
\]
If $L$ is the empty language, then $\alpha_L = 0$, and if $L = \Sigma^* $, then $\alpha_L = \frac{1}{3}$. For any other $L$, $\alpha_L$ takes a value in between.

\begin{theorem}
    \label{thm:q-unary}
    Any unary language $L \subseteq \{a\}^*$ is verified by a bounded-error MA-PostQFA $M$ using linear size binary certificates.
\end{theorem}
\begin{proof}
    For simplicity we focus on only the main part of state vector, and we omit the normalization factors as well as the auxiliary entries.

    After reading symbol $\lmoon$, the state vector is set to
    \[
        \myvector{1 \\ \alpha_L \\ 0 \\ \vdots \\ 0 }.
    \]
     Let $x = a^n$ be the given input. The expected certificate for $x$ is as
    \[
    F_L(1) F_L(2) \cdots F_L(n) F_L(n+1).
    \]
    While reading the provided certificate with length $k$, say $c_x$, the verifier $M$ compares the certificate digits with the digits of $\alpha_L$:
    \begin{itemize}
        \item After reading the first digit of $c_x$, $M$ updates the 2nd entry with $4\alpha_L - F_L(1)$, say $\alpha_1$.
        \item After reading the second digit of $c_x$, $M$ updates the 2nd entry with $4\alpha_1 - F_L(2)$, say $\alpha_2$.
        \item And so on. 
    \end{itemize}
    As long as the digits are correct, the integer part of $\alpha_i$s is 0. Otherwise, it will be non-zero.

    The verifier also compares the length of $c_x$ with the input length $x$, i.e., $ |c_x| \overset{?}{=} |x| + 1  $. Thus, while reading $c_x$, its length is encoded into the 3rd entry.

    After reading the certificate, the first three entries of state vector is 
    \[
        \myvector{1 \\ \alpha_k \\ k } .
    \]
    If the last digit of certificate is 0, then the input is rejected with probability 1. Remark that this is not possible when $a^n \in L$, as the prover provides a valid certificate. 

    In the remaining part, we assume that the last digit of certificate is 1. After reading $\cent a^n$, the value of $n+1$ is encoded on the fourth entry and so the first four entries of state vector is
    \[
        \myvector{1 \\ \alpha_k \\ k \\ n+1 } .
    \]
    
    After reading $\dollar$, the first three entries of state vector is
    \[
        \myvector{1 \\ 2\alpha_k \\ 2(k - n - 1) }.
    \]
    Here the first one is the accepting state and the second and third ones are rejecting states.

    If $a^n$ is in the language, $2\alpha_t$ is in $[0,2/3]$ and $(k-n-1) = 0$. Then, the accepting probability is at least
    \[
        \dfrac{1}{1+\mypar{2/3}^2} = \dfrac{9}{13} > 0.69 .   
    \]
    If $a^n$ is not in the language, then the certificate can be shorter or longer than expected.  That is $|2(k-n-1)| \geq 2$. Thus, regardless of the value of $2\alpha_k$, the input is rejected with probability at least
    \[
        \dfrac{2^2}{1+2^2} = \dfrac{4}{5} = 0.8 .
    \]
    If the certificate length is as expected, then $|2\alpha_k| \geq 2$. So, similarly, the input is rejected with probability at least 
    \[
        \dfrac{4}{5} = 0.8.
    \]
    
    The probabilities can be amplified in different ways. One trivial way is to execute several copies of the verifiers and make a postselection on all accepting and all rejecting states (the outcomes with some accepting and some rejecting states are discarded). 
\end{proof}

We obtain the same result also for binary languages but using exponential size certificate, presented in the next section. Regarding binary languages and linear size certificates, we present a bounded-error MA-PostQFA for NP-complete problem $\mathtt{SUBSETSUM}$:
\[
    \mathtt{SUBSETSUM} = \{ S \# b_1 \# \cdots 
 \# b_n \mid \exists I \subseteq \{1,\ldots,n\} \mbox{ and } S = \sum_{i \in I} b_i  \},
\]
where $S$ and $b_1,\ldots,b_n$ are positive binary numbers, each of which starts with 1.

\begin{theorem}
    \label{thm:subsetsum}
    Language $\mathtt{SUBSETSUM}$ is verified by a bounded-error MA-PostQFA $M$ using linear size binary certificate.
\end{theorem}
\begin{proof}
    The formats of input and certificate are checked deterministically, and the input is rejected with probability 1 if any defect is detected.

    For a given input $x \in \{0,1,\#\}$, the expected certificate is the input itself but with the indicators showing which $b_i$s are selected. We use double $\#$s if the successor $b_i$ is selected.

    The verifier $M$ encodes the certificate (by replacing double $\#$ with single $\#$) and input in base-4, and then compare them by subtracting. The result is 0 if they are identical, and the absolute value of result is at least $d_1 \geq 1$, otherwise. We use the following linear operator for encoding 0, 1, and $\#$ when starting in $(1~~0)^T$:
    \[
        A_0 = \mymatrix{cc}{1 & 0 \\ 1 & 4} , ~~
        A_1 = \mymatrix{cc}{1 & 0 \\ 2 & 4}
        ~\mbox{ and }~
        A_{\#} = \mymatrix{cc}{1 & 0 \\ 3 & 4},
    \]
    where $0$ is encoded as 1, $1$ is encoded as 2, and $\#$ is encoded as 3. Then, for example, the encoding of $10\#101$ is
    \[
        A_1 A_0 A_1 A_{\#} A_0 A_1 \myvector{1 \\ 0}= 
        \myvector{ 0 \\
        2\cdot 2^5 + 1\cdot 2^4 + 3 \cdot 2^3 + 2\cdot 2^2 + 1\cdot 2^1 + 2 }  . 
    \]
    The verifier $M$ also encodes $S$ and the summation of the selected $b_i$s in base-2, by using the following linear operators when started in $(1~~0)^T$:
    \[
        A_0 = \mymatrix{cc}{1 & 0 \\ 0 & 2} , ~~
        ~\mbox{ and }~
        A_1 = \mymatrix{cc}{1 & 0 \\ 1 & 2}        
    \]
    Then, the results are compared by subtracting. The result is 0 if there is indeed a correct selection of $b_i$s, and the absolute value of result is at least $d_2 \geq 1$ if there is no such subset.

    Thus, if the input is a member string, then it is accepted with probability 1. Otherwise, the input is rejected with high probability, i.e., at the end, the value of rejecting states are set to $ (t \cdot d_1) $ and $(t \cdot d_2)$ for some integer $t>1$. 
\end{proof}

\section{Verification with exponential size certificates}
\label{sec:exp}

Similar to verifying every unary language, we can verify every binary languages by MA-PostQFAs but using exponential size certificates. We adopt a proof again given in \cite{SY17}, in which it was shown that two-way QFAs (with classical head) can verify any binary languages with double exponential expected time.

Let $L \subseteq \{a,b\}^*$ be a binary language. We can put all binary stings in lexicographical order as 
\[
    \begin{array}{lcl}
       \Sigma^*[1]  & = & \varepsilon  \\
       \Sigma^*[2]  & = & a  \\
       \Sigma^*[3]  & = & b  \\
       \Sigma^*[4]  & = & ab  \\
       \Sigma^*[5]  & = & ba  \\
       \Sigma^*[6]  & = & bb  \\
       \Sigma^*[7]  & = & aaa  \\
       \vdots & \vdots & \vdots
    \end{array}
\]
We use the same function
\[
    F_L(i) = \left\lbrace
    \begin{array}{cc}
        0 & \mbox{if } \Sigma^*[i] \notin L \\
        1  & \mbox{if } \Sigma^*[i] \in L
    \end{array}
    \right .
\]
and the same format real number storing all membership information of $L$
\[
    \alpha_L = \sum_{i=1}^\infty \dfrac{F_L(i)}{4^i} = \dfrac{F_L(1)}{4} +  \dfrac{F_L(2)}{4^2} +  \dfrac{F_L(3)}{4^3} + \cdots .
\]

\begin{theorem}
    \label{thm:q-binary}
    Any binary language $L \subseteq \{a,b\}^*$ is verified by a bounded-error MA-PostQFA $M$ using exponential size binary certificates.
\end{theorem}
\begin{proof}
    Let $x \in \{a,b\}^*$ be the given input. We define the certificate with several symbols, but it can be easily converted to binary with linear increase in the size. 
    
    The expected certificate is composed by all strings until $x$ (included) in the lexicographic order with their membership bits: 
    \[
    \varepsilon \# F_L(1) \# a \# F_L(2) \# b \# F_L(3) \# aa \# F_L(4) \# ab \# F_L(5) \# \cdots \# x \# F_L(i_x),
    \]
    where $i_x$ is the index of $x$ in the lexicographic ordering of binary strings. If the certificate is not in this form, then it is rejected with probability 1.

    We replace $a$ with 0 and $b$ with 1, and so, encode each binary string in base-2 by also putting additional 1 as the leading digit, e.g., we encode $aaba$ as $10010$ and $bb$ as $111$. While reading the certificate, the verifier check whether the strings are in lexicographic order or not. For example, after 111, 0000 should come. We encode 111 as 1111 and also add 1. We encode the next string as $ 1y_{i+1} $. Then, we compare their equality by subtracting: if $ 1111+1 = 1y_{i+1} $. On the state vector, we use two entries for them:
    \[
        \myvector{ e_i + 1 = (2^3+2^2+2^1+2^0)+1 \\ e_{i+1} } = \myvector{16 \\ e_{i+1}},
    \]
    where $e_{i+1}$ is the encoding of $1y_{i+1}.$
    We check their difference by using one of the entries
    \[
        \myvector{16-e_{i+1}} = \myvector{d_{i}}.
    \]
    If $y_{i+1}$ is 0000, $e_{i+1}$ is also 16. Thus, $d_{i} = 0$. Otherwise, $|d_{i}|$ is at least 1. We do this check iteratively for all successive pairs on the certificate. We tensor the related part of state vector with itself to obtain $d^2_{i+1}$, which is always non-negative. (For this, we also tensor the part of quantum operators.) Then, we use one of the entries for the cumulative summation of these differences:
    \[
    \myvector{d_1^2+d_2^2+d_3^2+\cdots},
    \]
    which is 0 if there is no defect, and at least 1, otherwise. In this iteration, one additional check is made to compare the last string in the certificate, say $y_t$, with the input itself (x), as they are expected to be identical.

    As described in the proof of Theorem~\ref{thm:q-unary}, the provided membership bits (as a part of certificate) are compared with the digits of $\alpha_L$. Here, if the last symbol of the certificate is 0 (a membership bit), then input is rejected probability 1. The quantum decision (at the end) is made only if this bit is 1.
    
    At the end, we have the following non-postselecting values on the state vector:
    \[
        \myvector{1 \\ 2 \alpha_k \\ 2 ( d_1^2+d_2^2+\cdots+d_k^2 ) },
    \]
    where $k$ is the number of strings on the certificate including the empty string and $\alpha_k$ is the same value as described for the unary case. 
    
    The probability/error analysis here is the same as in the proof of Theorem~\ref{thm:q-unary}. So, the members are accepted with high probability, and the other strings are rejected with high probability.
\end{proof}

\section{Verification with arbitrarily long certificates}
\label{sec:arbitrary}

One-way finite automata with two counters (1D2CAs) can simulate the computation of any Turing machine (TM) on a given input \cite{Min67}. Therefore, they can recognize all and only recursive enumerable (Turing recognizable) languages. The computation history of a 1D2CA provided by a honest prover can be checked by bounded-error two-way PFAs \cite{CL89}. We obtain a similar result for MA-PostPFAs on unary decidable (recursive) languages.

\begin{theorem}
    Any decidable unary language $L \subseteq \{a\}^*$ can be verified by a bounded-error MA-PostPFA $M$.
\end{theorem}
\begin{proof}
    Let $M_L$ be a 1D2CA recognizing $L$. For a given input $x = a^n $, its configuration at step $t \geq 0$ can be represented as
    \[
        C_t = (state_t,\cent a^i,a^{n-i}\dollar,a^j,a^k), 
    \]
    where 
    \begin{itemize}
        \item $state_t$ is the state of $M_L$ at the $t$-th step,
        \item the input is divided into two as ``$\cent a^i$'' and ``$a^{n-i}\dollar$'' such that the input head is on the right-most symbol of ``$\cent a^i$'', and,
        \item the values of first and second counters are respectively $j$ and $k$.
    \end{itemize}
    The verifier $M$ simulates the computation of $M_L$ on the given input by asking a valid configuration history from Merlin:
    \[
        C_0 ~ C_1 ~ C_2 ~ \cdots ~ C_t ~ C_{t+1} ~ \cdots,
    \]
    where $C_0 = (state_0,\cent,a^n\dollar,,)$. Suppose that Merlin provides 
    \[
        C_0' ~ C_1' ~ C_2' ~ \cdots ~ C_t' ~ C_{t+1}' ~ \cdots.
    \]
    The verifier $M$ probabilistically checks whether the length of input provided in $C_0'$ is equal to the actual input length, and it deterministically checks the rest of details of $C_0'$. Then, $M$ checks whether $C_{t+1}'$ is a valid successor for $C_{t}'$ for every $t \geq 0$. By scanning $C_t'$, $M$ stores the current state, the symbol under the input head ($\cent$, $a$, or $\dollar$), and the status of the counters, i.e., whether the value of first/second counter is zero or nonzero. Then, based on the transition function of $M_L$, $M$ determines the followings:
    \begin{itemize}
        \item the next state,
        \item the input head is moved one symbol to the right or not,
        \item the value of the first counter is increased or decreased by 1 or does not change, and,
        \item the value of the second counter is increased or decreased by 1 or does not change.
    \end{itemize}
    Here the length of input parts and counter values differ by at most one.
    Thus, $M$ makes unary equality checks between $C_t'$ and $C_{t+1}'$ to determine whether $C_{t+1}'$ is a valid successor of $C_t'$. 
    
    If all checks are valid, then $M$ makes its decision on the input based on the decision of $M_L$. Otherwise, the input is rejected.
\end{proof}

This result can be extended for recursive enumerable unary languages by using ``weak'' MA-PostPFAs. (We can also use MA-PostPFAs indeed, but then we force the provers to provide finite certificates while the actual computation does not halt.)

\section{Concluding remarks}
\label{sec:conc-remark}

We introduce Merlin-Arthur automata where the certificate provided by Merlin is read once by Arthur before processing the input. This set-up can be seen as a limited nondeterminism: all nondeterministic choices is made before processing the input.

We show that Merlin-Arthur deterministic automata cannot verify any nonregular language, but they can be state efficient. 

We leave them open if bounded-error Merlin-Arthur probabilistic \& quantum automata can verify any nonregular language or they can be state efficient compared to their standalone models. But, we show that they can verify nonstochastic languages with cutpoints by using sublinear size certificates.

With the capability of postselection, their verification power is increased dramatically in the bounded-error case. Quantum automata can verify every unary language with linear size certificates and every binary language with exponential size certificates. We have weaker results for probabilistic automata: they can verify every decidable unary language but without any bound on the certificate size.

\section*{Acknowledgments}

We thank Viktor Olej\'ar for kindly sharing the references on multiple-entry DFAs, Kai Salomaa for kindly providing a copy of \cite{HolzerSY01}, and Mika Hirvensalo for kindly answering our question on nonstochastic unary languages.

Yakary{\i}lmaz was partially supported by the Latvian Quantum Initiative under European Union Recovery and Resilience Facility project no. 2.3.1.1.i.0/1/22/I/CFLA/001, the ERDF project Nr. 1.1.1.5/19/A/005 ``Quantum computers with constant memory'' and the project ``Quantum algorithms: from complexity theory to experiment'' funded under ERDF programme 1.1.1.5.

\bibliographystyle{plain}
\bibliography{ref}

\end{document}